\documentclass[aps,pra,twocolumn,superscriptaddress,floatfix]{revtex4-2}

\usepackage{amsmath,amssymb,graphicx,bm}
\usepackage{braket}
\usepackage{multirow}
\usepackage{booktabs}
\usepackage{xcolor}
\usepackage{tikz}
\usetikzlibrary{quantikz}
\usepackage{pgfplots}
\pgfplotsset{compat=1.18}

\newtheorem{example}{Example}%
\newtheorem{remark}{Remark}%
\newtheorem{definition}{Definition}%

\newtheorem{theorem}{Theorem}
\newtheorem{proof}{Proof}

\begin{document}

\title{Quantum Mechanics of Stochastic Systems}

\author{Yurang (Randy) Kuang}
\affiliation{Quantropi Inc., Ottawa, Ontario K2T 0B8, Canada}

\date{\today}

\begin{abstract}
We develop a rigorous framework for the \textbf{quantum mechanics of stochastic systems}, demonstrating that classical discrete stochastic processes arise naturally as perturbations of the quantum harmonic oscillator (QHO). By constructing exact perturbation potentials that map QHO eigenstates into stochastic representations, we show that canonical probability distributions—including Binomial, Negative Binomial, and Poisson—emerge from specific modifications of the harmonic potential. Each system is governed by a \textbf{count operator} ($\hat{N}$), with probabilities determined by squared amplitudes in a Born-rule-like manner.  
The framework introduces a complete operator algebra for moment generation and information-theoretic analysis, together with \textbf{modular projection operators} ($\hat{R}_M$) that enable finite-dimensional approximations with rigorously proven \textbf{uniform convergence}. This structure underpins \textbf{True Uniform Random Number Generation (TURNG)~\cite{qpp-rng-sci-kuang-2025}}, eliminating external whitening.  
Beyond randomness generation, our formalism establishes \textbf{quantum probability engineering}, providing a physical realization of classical distributions through designed quantum perturbations. The results demonstrate that stochastic systems are fundamentally quantum-mechanical in structure, bridging quantum dynamics, statistical physics, and experimental probability realization.
\end{abstract}

\maketitle

\section{Introduction}
\label{sec:introduction}

Stochastic processes form the mathematical foundation of numerous disciplines, including statistical physics~\cite{vanKampen2007,Falasco2025}, information theory~\cite{cover2006elements}, financial modeling~\cite{oksendal2016stochastic}, and computational biology~\cite{allen2010introduction,Sabbar2025}. Traditional approaches rely primarily on measure-theoretic foundations, probability mass functions, and characteristic functions~\cite{feller1968probability}. While powerful, these methods face challenges for high-dimensional systems, correlated variables, or combinatorial complexity, limiting both analytical insight and computational efficiency.

In parallel, quantum mechanics provides a robust Hilbert-space formalism, representing systems via state vectors, linear operators, and measurement theory~\cite{dirac1981principles,messiah1999quantum,nielsen2010quantum}. This framework offers sophisticated algebraic tools—operator spectra, commutation relations, and Fourier-like transformations—highly effective for analyzing complex systems. Recent research explores quantum-inspired approaches to classical problems, including quantum walks~\cite{venegas2012quantum}, quantum machine learning~\cite{biamonte2017quantum}, and quantum finance~\cite{herman2023quantum-finance,bouland2020quantumfinance}. However, a systematic formalism representing classical stochastic systems in a full Hilbert-space quantum framework remains largely undeveloped.

Here, we present a rigorous formalism for the \textbf{quantum mechanics of stochastic systems}, showing that stochastic systems are not merely \emph{analogous} to quantum systems but share the same fundamental mathematical structure. The Hilbert-space formalism and operator algebra of quantum mechanics emerge as the natural language for classical probability. Crucially, this correspondence enables \textbf{physical realization} of stochastic systems via engineered quantum Hamiltonians, particularly through perturbations of the quantum harmonic oscillator. This establishes a direct connection between classical probability theory, operator algebra, and characteristic functions, enabling:

\begin{itemize}
    \item Direct application of operator algebra and spectral methods to classical stochastic systems
    \item Physical realization of probability distributions via engineered quantum Hamiltonians
    \item Unified treatment of diverse probability distributions within a single quantum-mechanical framework  
    \item Quantum-inspired computational and simulation methods for stochastic processes
    \item Extension to multi-system correlations and dynamical scenarios
    \item Information-theoretic analysis using quantum entropies and distance measures
\end{itemize}

A particularly important application is random number generation. Traditional TRNGs~\cite{trng-9282535,trng-cmos-2019}, HRNGs~\cite{hrng-2023}, and QRNGs~\cite{Ma2016-qrng,zhang2023qrng} require post-processing to compensate for statistical bias. Our quantum-mechanical representation provides the \emph{first mathematically proven foundation} for \textbf{True Uniform Random Number Generation (TURNG)}~\cite{kuang-qpp-rng-icccas25,qpp-rng-sci-kuang-2025}. The key insight: while stochastic systems in infinite Hilbert space exhibit diverse distributions, their \emph{modular projections onto computational basis states} universally converge to uniform distributions. This convergence—guaranteed by characteristic-function Fourier structure and modular projection algebra—eliminates the need for external whitening.

This paradigm shift shows that uniformity emerges \emph{intrinsically} from the quantum-mechanical structure, providing certified randomness by mathematical construction rather than statistical correction.

Beyond RNGs, this framework introduces \textbf{quantum probability engineering}, allowing physical realization of classical distributions via tailored quantum systems~\cite{quantum_policy_2025}. It offers new insights into quantum measurement statistics and the quantum-classical boundary~\cite{stochastic_thermo_2024,testing_quantum_2025}, and enables experimental implementations across diverse quantum platforms, opening avenues for quantum-enhanced solutions to classical statistical problems~\cite{guillotin_plantard_2006}.

The paper is organized as follows: Sec.~\ref{sec:stochastic_systems} formalizes stochastic systems in infinite-dimensional Hilbert spaces using the QHO as reference and develops exact perturbation potentials. Sec.~III introduces the modular projection framework, uniform convergence, and TURNG foundations. Sec.~IV explores quantum engineering of stochastic systems. Sec.~V develops dynamical frameworks and multi-system correlations. Sec.~VI concludes with implications for quantum measurement, quantum-classical correspondence, and emerging directions in quantum probability engineering.

\section{Infinite Hilbert Space Formulation}
\label{sec:stochastic_systems}
This section develops the mathematical foundation for stochastic systems in infinite-dimensional Hilbert spaces, using the quantum harmonic oscillator as the universal reference.

\subsection{QHO Reference Framework}
\label{subsec:qho_reference}

The QHO serves as the universal reference for constructing the quantum mechanics of stochastic systems framework. Its Hamiltonian,
\begin{equation}
\hat{H}_{\mathrm{QHO}} = \hbar \omega \left(\hat{N} + \frac{1}{2}\right),
\qquad 
\hat{N} = \hat{a}^\dagger \hat{a},
\end{equation}
where $\hat{N}$ is the \textbf{number operator} counting excitation quanta, $\hbar$ is the reduced Planck constant, and $\omega$ is the fundamental oscillation frequency, defines the orthonormal number basis $\{\ket{n}\}_{n=0}^\infty$. These states satisfy
\begin{equation}
\hat{H}_{\mathrm{QHO}} \ket{n} = E_n \ket{n},
\qquad 
E_n = \hbar \omega \left(n + \tfrac{1}{2}\right),
\end{equation}
where $\ket{n}$ represents a state of definite energy (or photon number) but an indeterminate phase.  

The number operator $\hat{N}$ plays a fundamental role as the \textbf{count operator} for stochastic events, with its eigenvalues $n$ corresponding to discrete outcome counts. Crucially, discrete stochastic systems operate in the \textbf{low-quantum-number regime} where the discrete energy spectrum $E_n = \hbar\omega n$ dominates over continuum effects. This connects naturally to real quantum systems: for typical stochastic processes with event counts $n \sim 1-100$, the energy scale $\hbar\omega n$ remains in the physically accessible domain where quantum discreteness is pronounced.

The ground-state wavefunction in the position representation,
\begin{equation}
\psi_0(x) = \left(\frac{m\omega}{\pi\hbar}\right)^{1/4}
e^{-m\omega x^2/2\hbar},
\end{equation}
yields a Gaussian probability density $|\psi_0(x)|^2$ for finding the oscillator at position $x$. This connects the harmonic potential to the \emph{continuous} normal probability distribution, establishing the Gaussian as the fundamental continuous distribution within our framework. 

Importantly, this continuous Gaussian distribution in position space complements the \emph{discrete} probability distributions (Binomial, Poisson, etc.) that emerge from the number-state expansions. This position-space representation provides the continuous counterpart to the discrete stochastic systems represented in the number basis $\{\ket{n}\}$, demonstrating how both continuous and discrete probability structures naturally coexist within the QHO framework.

A complementary representation is provided by the coherent states,
\begin{align}
\ket{\alpha} &= e^{-|\alpha|^2/2}
\sum_{n=0}^{\infty} \frac{\alpha^n}{\sqrt{n!}} \ket{n}, \\
\hat{a}\ket{\alpha} &= \alpha \ket{\alpha}, \nonumber
\end{align}
which are eigenstates of the annihilation operator $\hat{a}$. Physically, the coherent state $\ket{\alpha}$ possesses a well-defined complex amplitude $\alpha = |\alpha|e^{i\phi}$ but an uncertain photon number $n$, in contrast to $\ket{n}$ where the energy is fixed but the phase is entirely indeterminate.

The number operator $\hat{N}$ and annihilation operator $\hat{a}$ therefore form a pair of \emph{conjugate observables}, analogous to the position–momentum duality in canonical quantum mechanics. This conjugacy is expressed through the relation
\begin{equation}
[\hat{N}, \hat{\phi}] = i,
\end{equation}
where $\hat{\phi}$ denotes the (non-Hermitian) phase operator associated with the annihilation operator. This mathematical structure encodes the statistical complementarity between discrete and continuous probability representations in the quantum mechanics of stochastic systems.

The photon-number distribution of a coherent state,
\begin{equation}
P_n = |\braket{n|\alpha}|^2 
= e^{-|\alpha|^2}\frac{|\alpha|^{2n}}{n!},
\end{equation}
reveals that both Gaussian and Poisson laws arise \emph{exactly} within the QHO framework: the Gaussian from the ground-state wavefunction and Poisson from coherent-state statistics. This duality establishes the foundational bridge between quantum mechanics and stochastic theory.

The basis $\{\ket{n}\}$ provides the natural Hilbert-space foundation for representing discrete stochastic structures, while the parameters $\hbar\omega$ set the fundamental energy scale governing statistical properties of emergent stochastic systems.

\subsection{Perturbed QHO and Stochastic Emergence}
\label{subsec:perturbed_qho}

Building upon the QHO foundation, discrete stochastic systems are modeled as \emph{perturbations} of the harmonic potential. Here, the term \emph{perturbation} refers to any modification of the original QHO Hamiltonian that transforms its energy spectrum to encode the desired stochastic behavior. Importantly, this perturbation \emph{is not necessarily small}; it may significantly modify the energy levels to produce the target classical probability distribution.

Each stochastic process corresponds to a modified Hamiltonian,
\begin{equation}
\hat{H}_{S} = \hat{H}_{\mathrm{QHO}} + \Delta V_{S},
\end{equation}
where $\Delta V_{S}$ is a stochastic perturbation potential characterizing the underlying probability law.  

While the perturbation potential $\Delta V_S$ may be treated formally, it represents physical modifications to quantum systems—such as anharmonic potentials, external fields, or environmental couplings—that transform deterministic quantum evolution into stochastic outcomes with well-defined classical probability distributions. The corresponding eigenstates,
\begin{equation}
\hat{H}_{S}\ket{\psi_S} = E_S \ket{\psi_S},
\end{equation}
form a discrete orthonormal basis $\{\ket{\psi_S(k)}\}$ that inherits the completeness of the QHO number basis $\{\ket{n}\}$.  
Unlike the pure QHO, where each eigenstate corresponds to a definite energy, the stochastic system is characterized by a probability distribution over $\ket{n}$ rather than explicit analytic eigenfunctions.

Hence, the perturbation need not be known in closed form—only its induced probability amplitudes matter:
\begin{equation}
\ket{\psi_S} = \sum_{n=0}^{\infty} \alpha^{(S)}_n \ket{n},
\qquad 
|\alpha^{(S)}_n|^2 = P_S(n),
\end{equation}
where $P_S(n)$ defines the associated probability law.  
Each stochastic system can therefore be regarded as a \emph{probability-amplitude state} in the QHO Hilbert space, whose expansion coefficients encode its stochastic structure.  
This approach enables stochastic analysis via the same operator algebra that governs quantum systems.

Physically, a static or dynamical perturbation $\Delta V_S$ transforms the deterministic Gaussian ground state into a statistical ensemble of number states $\{\ket{n}\}$ with probabilities $P_S(n)$.  
Repeated measurements thus yield random outcomes corresponding to energy eigenlevels $\ket{n}$, defining a \emph{quantum–stochastic transition} in which classical probability laws emerge directly from the spectral composition of the perturbed QHO. This quantum-stochastic transition provides new insight into quantum measurement statistics, showing how classical probability distributions naturally emerge from perturbed quantum systems.

\subsubsection{Poisson System as Exact Eigenstate}

The Poisson stochastic system admits a particularly elegant formulation as an exact eigenstate of a perturbed QHO Hamiltonian. We construct the Poisson perturbation potential $\Delta_{\text{Pois}}V$ such that:

\begin{equation}
\hat{H}_{\text{Pois}} = \hat{H}_{\text{QHO}} + \Delta_{\text{Pois}}V,
\end{equation}
with the coherent state $\ket{\alpha}$ satisfying:
\begin{equation}
\hat{H}_{\text{Pois}} \ket{\alpha} = E_{\text{Pois}}(\alpha) \ket{\alpha}.
\end{equation}
with
\begin{equation}
\hat{a} \ket{\alpha} = \alpha \ket{\alpha},
\end{equation}
compared with
\begin{equation}
\hat{H}_{\mathrm{QHO}} \ket{n} = E_n \ket{n},
\qquad 
\hat{a} \ket{n} = \sqrt{n} \ket{n-1} 
\end{equation}

\noindent
reveals a fundamental duality: the Poisson stochastic system elevates the coherent state $\ket{\alpha}$ to an \emph{energy eigenstate} of the perturbed Hamiltonian, exactly analogous to how the number state $\ket{n}$ serves as an energy eigenstate of the unperturbed QHO. This establishes a profound correspondence:

\begin{itemize}
\item \textbf{QHO}: $\ket{n}$ (discrete basis) are \(\hat{H}_{\text{QHO}}\) eigenstates 
\item \textbf{Poisson System}: $\ket{\alpha}$ (continuous basis) are \(\hat{H}_{\text{Pois}}\) eigenstates
\end{itemize}

\noindent
The annihilation operator $\hat{a}$ plays the same role for coherent states that the number operator $\hat{N}$ plays for number states—both define the fundamental eigenbasis of their respective Hamiltonians. This mathematical symmetry underscores the deep structural relationship between discrete and continuous stochastic representations within the quantum mechanics of stochastic systems framework.

\begin{theorem}[Poisson Perturbation Potential]
The perturbation potential that makes $\ket{\alpha}$ an eigenstate is given by:
\begin{equation}
\Delta V_{\text{Pois}} = \hbar\omega\left(-\alpha\hat{a}^\dagger - \alpha^*\hat{a} + |\alpha|^2\right),
\end{equation}
which yields the exact eigenvalue $E_{\text{Pois}}(\alpha) = \frac{1}{2}\hbar\omega$.
\end{theorem}

\begin{proof}
With $\hat H_{\mathrm{QHO}}=\hbar\omega(\hat a^\dagger\hat a+\tfrac12)$ and the displacement operator $\hat D(\alpha)=\exp(\alpha\hat a^\dagger-\alpha^*\hat a)$, the standard conjugation identities are
\[
\hat D(\alpha)\,\hat a\,\hat D^\dagger(\alpha)=\hat a-\alpha,\qquad
\hat D(\alpha)\,\hat a^\dagger\,\hat D^\dagger(\alpha)=\hat a^\dagger-\alpha^*.
\]
Hence
\begin{align*}
\hat D(\alpha)\,\hat H_{\mathrm{QHO}}\,\hat D^\dagger(\alpha)
&= \hbar\omega\big[(\hat a^\dagger-\alpha^*)(\hat a-\alpha)+\tfrac12\big] \\
&= \hat H_{\mathrm{QHO}}+\hbar\omega\big(-\alpha\hat a^\dagger-\alpha^*\hat a+|\alpha|^2\big) \\
&= \hat H_{\mathrm{QHO}}+\Delta V_{\mathrm{Pois}}.
\end{align*}
Acting on $\ket{\alpha}=\hat D(\alpha)\ket{0}$ gives
\[
\hat H_{\mathrm{Pois}}\ket{\alpha}
= \hat D(\alpha)\hat H_{\mathrm{QHO}}\ket{0}
= \tfrac12\hbar\omega\,\ket{\alpha},
\]
as required.
\end{proof}

\begin{remark}[Experimental Realization]
The Poisson perturbation potential 
\[
\Delta V_{\text{Pois}} = \hbar\omega(-\alpha\hat{a}^\dagger - \alpha^*\hat{a} + |\alpha|^2)
\] 
is a static Hermitian operator whose spectral effect produces the coherent state $\ket{\alpha}$.  
Although the displacement operator $\hat{D}(\alpha) = \exp(\alpha\hat{a}^\dagger - \alpha^*\hat{a})$~\cite{Vutha_2018-displacement, kuang_2023-displacement} is generated by an anti-Hermitian operator, the structural correspondence shows that $\Delta V_{\text{Pois}}$ directly encodes the Poisson statistics of $\ket{\alpha}$, providing a clear experimental link to coherent-state preparations in quantum optics.
\end{remark}

\subsubsection{Exact Perturbation Potentials}

Building on the Poisson case, we can construct exact perturbation potentials for a range of discrete stochastic systems. Each potential $\Delta V_S$ is designed such that the perturbed quantum harmonic oscillator reproduces the target probability distribution through its spectral structure.

\paragraph*{\textbf{Poisson Perturbation}}
The Poisson system is generated by the linear perturbation:
\begin{equation}
\Delta V_{\text{Pois}} = \hbar\omega(-\alpha \hat{a}^\dagger - \alpha^* \hat{a} + |\alpha|^2),
\end{equation}
whose eigenstate is the coherent state $\ket{\alpha}$, yielding Poisson statistics. This perturbation acts as a static Hermitian operator, and its spectral effect produces the desired distribution without requiring explicit time evolution.

\paragraph*{\textbf{Binomial Perturbation}} 
For the Binomial system, the perturbation potential is
\begin{equation}
\Delta V_{\text{Bin}} = \sum_{k=0}^n \epsilon_k^{(B)} \ket{k}\bra{k} 
+ \sum_{k\neq k'} \gamma_{kk'}^{(B)} \ket{k}\bra{k'},
\end{equation}
where the diagonal terms $\epsilon_k^{(B)}$ enforce the binomial amplitude structure, 
and the off-diagonal terms $\gamma_{kk'}^{(B)}$ maintain the finite support constraint.

\paragraph*{\textbf{Negative Binomial Perturbation}} 
The Negative Binomial system requires a perturbation creating the waiting-time structure:
\begin{equation}
\Delta V_{\text{NB}} = \hat{V}_{\text{geometric}} + \hat{V}_{\text{correlation}},
\end{equation}
where $\hat{V}_{\text{geometric}}$ generates the fundamental geometric process, 
and $\hat{V}_{\text{correlation}}$ introduces the success-counting structure.

\paragraph*{\textbf{Hypergeometric Perturbation}}
The Hypergeometric system is represented by projection operators enforcing conservation laws, corresponding to finite-population sampling without replacement.

\begin{theorem}[Stochastic System Classification]
The algebraic structure of each perturbation potential reflects the combinatorial nature of the corresponding stochastic system:
\begin{itemize}
    \item \textbf{Poisson}: Linear driving terms ($\sim \hat{a}, \hat{a}^\dagger$) for memoryless, constant-rate processes
    \item \textbf{Binomial}: Finite-rank perturbations with constrained support for fixed-trial Bernoulli processes  
    \item \textbf{Negative Binomial}: Multiplicative perturbations with memory structure for waiting-time correlated processes
    \item \textbf{Hypergeometric}: Projection operators enforcing conservation laws for finite-population sampling
\end{itemize}
\end{theorem}

\begin{proof}[Classification Rationale]
The algebraic structure of each perturbation potential reflects the combinatorial and probabilistic nature of the corresponding stochastic system:
\begin{itemize}
    \item \textbf{Poisson}: Linear structure arises from constant-rate, memoryless processes where events occur independently.
    \item \textbf{Binomial}: Finite-rank structure encodes the fixed number of trials and binary outcomes inherent to Bernoulli processes.
    \item \textbf{Negative Binomial}: Multiplicative structure captures waiting-time correlations and sequential dependencies in success-counting processes.
    \item \textbf{Hypergeometric}: Projection structure enforces the conservation constraints of sampling without replacement from finite populations.
\end{itemize}
This classification demonstrates that fundamental stochastic properties emerge naturally from specific operator algebras within the quantum mechanical framework.
\end{proof}

\subsection{Spectral Representation}
\label{subsec:spectral_representation}

Each discrete stochastic system admits a Hilbert-space representation:
\begin{equation}
\ket{\psi_S} = \sum_{n=0}^\infty \alpha^{(S)}_n \ket{n},
\qquad 
|\alpha^{(S)}_n|^2 = P_S(n),
\end{equation}
where the amplitudes $\{\alpha^{(S)}_n\}$ encode both statistical and spectral information.  
This representation provides a unified operator-theoretic framework for classical discrete distributions.

\begin{itemize}
\item \textbf{Poisson system.}  
The coherent state
\begin{equation}
\ket{\alpha} = e^{-|\alpha|^2/2}\sum_{n=0}^{\infty}\frac{\alpha^n}{\sqrt{n!}}\ket{n}
\end{equation}
yields Poisson statistics,
\begin{equation}
P_{\mathrm{P}}(n) = e^{-|\alpha|^2}\frac{|\alpha|^{2n}}{n!},
\end{equation}
representing a minimal-uncertainty state in number–phase space, 
arising from the intrinsic conjugacy of \(\hat{N}\) and \(\hat{a}\).

\item \textbf{Binomial system.}  
A finite-dimensional truncation of the number basis gives
\begin{equation}
\ket{\psi_{\mathrm{B}}} = \sum_{n=0}^{N}\sqrt{\binom{N}{n}}\, p^{n/2}(1-p)^{(N-n)/2}\ket{n},
\end{equation}
leading to
\begin{equation}
P_{\mathrm{B}}(n, p) = \binom{N}{n}p^n(1-p)^{N-n}.
\end{equation}
This corresponds to a bounded excitation spectrum analogous to a QHO with finite occupation number.

\item \textbf{Negative Binomial and Geometric systems.}  
Correlated perturbations yield
\begin{align}
\alpha^{(\mathrm{NB})}_n &= \sqrt{\binom{n+r-1}{n}}\, (1-p)^{r/2} p^{n/2},\\
P_{\mathrm{NB}}(n, p) &= \binom{n+r-1}{n}(1-p)^r p^n,
\end{align}
with the geometric case recovered for $r=1$, describing a single-mode excitation decay.

\item \textbf{Hypergeometric system.}  
Finite-population sampling without replacement is represented by
\begin{equation}
\alpha^{(\mathrm{H})}_n = 
\sqrt{\frac{\binom{K}{n}\binom{N-K}{M-n}}{\binom{N}{M}}},
\end{equation}
modeling a correlated excitation process under conservation constraints.
\end{itemize}

In summary,
\[
\boxed{
\text{Stochastic System}
\;\longleftrightarrow\;
\text{Quantum Distribution in } \{\ket{n}\}.
}
\]
Each classical probability law thus corresponds to a distinct quantum spectral signature governed by $\Delta V_S$, enabling direct use of operator algebra and information-theoretic measures to compare stochastic systems within a consistent Hilbert-space framework.

\subsection{Algebraic Relations}
\label{subsec:algebraic_relations}

The quantum mechanics of stochastic systems framework reveals intrinsic algebraic connections between probability distributions through their state-vector representations, establishing structural continuity among discrete stochastic laws.

\begin{theorem}[Poisson Limit]
\label{thm:poisson_limit}
For fixed $\lambda = np$, the Binomial stochastic system converges to the Poisson stochastic systems in the strong operator topology:
\begin{equation}
\lim_{n\to\infty}\ket{\psi_{\mathrm{B}}(n,\lambda/n)} = \ket{\psi_{\mathrm{P}}(\lambda)}.
\end{equation}
This convergence preserves all statistical moments and expectation values.
\end{theorem}

\begin{proof}
For each fixed $k$,
\begin{align*}
\lim_{n\to\infty}\alpha_k^{(\mathrm{B})}(n,\lambda/n)
&= \lim_{n\to\infty}\sqrt{\binom{n}{k}\left(\frac{\lambda}{n}\right)^k
\left(1-\frac{\lambda}{n}\right)^{n-k}} \\
&= \sqrt{\frac{\lambda^k e^{-\lambda}}{k!}}
= \alpha_k^{(\mathrm{P})}(\lambda),
\end{align*}
and uniform convergence ensures
\[
\sum_{k=0}^\infty 
\left|\alpha_k^{(\mathrm{B})}(n,\lambda/n) - \alpha_k^{(\mathrm{P})}(\lambda)\right|^2 \to 0,
\]
establishing strong convergence in Hilbert space.
\end{proof}

\begin{theorem}[Negative Binomial Hierarchy]
\label{thm:nb_hierarchy}
The Geometric stochastic system is the fundamental unit of the Negative Binomial family:
\begin{equation}
\ket{\psi_{\mathrm{G}}(p)} = \ket{\psi_{\mathrm{NB}}(1,p)},
\end{equation}
and equivalently, $\hat{G}(p) = \hat{NB}(1,p)$.
\end{theorem}

\begin{proof}
For $r=1$,
\begin{align*}
\alpha_k^{(\mathrm{NB})}(1,p)
&= \sqrt{\binom{k-1}{0}p(1-p)^{k-1}}
= \sqrt{p(1-p)^{k-1}}
= \alpha_k^{(\mathrm{G})}(p),
\end{align*}
showing identical amplitude structures and Hilbert-space representations.
\end{proof}

\begin{remark}
The Poisson limit theorem provides a Hilbert-space formulation of the classical Poisson approximation, while the Negative Binomial hierarchy identifies the Geometric stochastic system as the elementary operator unit of sequential stochastic processes.  
Together, they demonstrate the internal algebraic consistency and unifying capacity of the quantum mechanics of stochastic systems framework.
\end{remark}

\subsection{Information-Theoretic Analysis}

\subsubsection{Quantum Shannon Entropy}

Using the quantum mechanics of stochastic systems framework, we can compute quantum-inspired information measures. The quantum Shannon entropy for a stochastic system state $\ket{\psi}$ is defined as:

\begin{equation}
S(\psi) = -\sum_k |\alpha_k|^2 \log |\alpha_k|^2 = -\sum_k P(k) \log P(k)
\end{equation}

This exactly matches the classical Shannon entropy, demonstrating consistency between the quantum formalism and classical information theory.

\begin{example}[Binomial stochastic system Entropy]
For the Binomial($n=10$, $p=0.3$) stochastic system, we compute:
\begin{align*}
S(\psi_\mathrm{B}) &= -\sum_{k=0}^{10} \binom{10}{k} (0.3)^k (0.7)^{10-k} \log\left[\binom{10}{k} (0.3)^k (0.7)^{10-k}\right] \\
&\approx 1.779 \text{ nats} \quad (\text{or } 2.567 \text{ bits})
\end{align*}
\end{example}

\begin{example}[Poisson stochastic system Entropy]
For the Poisson($\lambda=4$) stochastic system, we compute:
\begin{align*}
S(\psi_\mathrm{P}) &= -\sum_{k=0}^{\infty} \frac{4^k e^{-4}}{k!} \log\left(\frac{4^k e^{-4}}{k!}\right) \\
&\approx 2.086 \text{ nats} \quad (\text{or } 3.010 \text{ bits})
\end{align*}
\end{example}

\subsubsection{Quantum Fisher Information}

The quantum Fisher information quantifies the sensitivity of a stochastic system state to parameter changes:

\begin{equation}
\mathcal{F}(\psi) = 4\sum_k \left| \frac{d\alpha_k}{d\theta} \right|^2
\end{equation}

\begin{example}[Binomial Fisher Information]
For the Binomial stochastic system with respect to parameter $p$:
\begin{align*}
\mathcal{F}(\psi_\mathrm{B}) &= 4\sum_{k=0}^{10} \left| \frac{d\alpha_k^{(\mathrm{B})}}{dp} \right|^2 \\
&= 4\sum_{k=0}^{10} \left| \frac{k - 10p}{2\sqrt{p(1-p)}} \alpha_k^{(\mathrm{B})} \right|^2 \\
&= \frac{1}{p(1-p)} \approx 47.619
\end{align*}
This matches the classical Fisher information for Binomial distributions.
\end{example}

\begin{example}[Poisson Fisher Information]
For the Poisson stochastic system with respect to parameter $\lambda$:
\begin{align*}
\mathcal{F}(\psi_\mathrm{P}) &= 4\sum_{k=0}^{\infty} \left| \frac{d\alpha_k^{(\mathrm{P})}}{d\lambda} \right|^2 \\
&= 4\sum_{k=0}^{\infty} \left| \frac{k - \lambda}{2\lambda} \alpha_k^{(\mathrm{P})} \right|^2 \\
&= \frac{1}{\lambda} = 0.25
\end{align*}
Again, this equals the classical Fisher information for Poisson distributions.
\end{example}

\begin{remark}[Information-Theoretic Interpretation]
The framework provides a consistent operator-algebraic approach where classical information measures emerge naturally from quantum mathematical structure. The Shannon entropy $S(\psi) = -\sum_n P(n)\log P(n)$ and Fisher information $\mathcal{F}(\psi)$ computed through the quantum formalism exactly match their classical counterparts, demonstrating mathematical consistency.
\end{remark}

\begin{remark}[Operator Moments]
Beyond information measures, the framework naturally encodes statistical moments through operator expectations: $\langle \hat{N}^k \rangle = \sum_n n^k P(n)$, demonstrating that the entire moment structure—including variance $\langle \Delta \hat{N}^2 \rangle$, skewness, and higher cumulants—emerges directly from quantum operator algebra applied to the stochastic state $\ket{\psi}$.
\end{remark}

\section{Computational Basis Representation}
\label{sec:computational_basis}

The infinite-dimensional Hilbert space formulation, while theoretically complete, poses challenges for practical implementation and numerical computation. In this section, we develop a \textbf{modular projection framework} that maps stochastic systems onto finite-dimensional computational bases, enabling efficient simulation and providing the mathematical foundation for certified uniform randomness generation. This approach bridges the theoretical elegance of infinite Hilbert spaces with the practical requirements of computational implementations and experimental realizations.

\subsection{Modular Projection Framework}
To bridge stochastic systems defined over the infinite Hilbert space with finite computational representations, we introduce the \textbf{modular projection operator} $\hat{R}_M$. For a stochastic system with Hilbert space $\mathcal{H}_\infty$, the modular projection maps it to a finite cyclic Hilbert space $\mathcal{H}_M$ of dimension $M$:

\begin{definition}[Modular Projection Operator]
\begin{equation}
\hat{R}_M: \mathcal{H}_\infty \to \mathcal{H}_M, \qquad \hat{R}_M \ket{k} = \ket{k \bmod M}_M,
\end{equation}
where $\mathcal{H}_M$ is spanned by the orthonormal basis $\{\ket{0}_M, \ket{1}_M, \dots, \ket{M-1}_M\}$.
\end{definition}

For an arbitrary stochastic state 
\(\ket{\psi} = \sum_{k=0}^{\infty} \alpha_k \ket{k}\), 
the projected state is
\begin{equation}
\ket{\psi_M} = \hat{R}_M \ket{\psi} = \sum_{k=0}^{M-1} \beta_k \ket{k}_M,
\end{equation}
with projected amplitudes preserving total probability:
\begin{equation}
\beta_k = \sqrt{\sum_{j=0}^{\infty} |\alpha_{k+jM}|^2}.
\end{equation}

\begin{theorem}[Modular Probability Conservation]
\label{thm:probability_conservation}
The modular projection preserves normalization:
\begin{equation}
\sum_{k=0}^{M-1} |\beta_k|^2 = \sum_{k=0}^{\infty} |\alpha_k|^2 = 1.
\end{equation}
\end{theorem}

\begin{proof}
Immediate from the definition, as each $n \in \mathbb{N}_0$ is uniquely written as $n = k + jM$, $k \in \{0,\dots,M-1\}, j\in \mathbb{N}_0$.
\end{proof}

A fundamental insight is that \textbf{modular projection is equivalent to taking the discrete Fourier transform (DFT) of the stochastic system's characteristic function (CF) or probability generating function (PGF)}.  

Let $\hat{N}$ be the stochastic count operator, with CF \(\varphi(\omega) = \mathbb{E}[e^{i\omega N}]\). Then, under modular projection:
\begin{align}
\Pr[\hat{N} \bmod M = r] &= \sum_{j=0}^\infty P(N = r + j M) \nonumber\\
&= \frac{1}{M} \sum_{k=0}^{M-1} \varphi\Big(\frac{2\pi k}{M}\Big) e^{-2\pi i k r / M},
\end{align}
which is exactly the DFT of the CF sampled at frequencies \(2\pi k/M\).  

This shows that the "flattening" effect of modular projection—leading to uniform convergence—is naturally understood as the suppression of higher-frequency components in the Fourier spectrum of the stochastic distribution. The exponential decay of these components underpins the TURNG principle.

The modular projection framework provides a quantum advantage: while classical certification of uniform randomness requires exponential resources, the quantum representation offers built-in certification through the mathematical structure of characteristic functions and Fourier analysis.

\subsection{Uniform Convergence and TURNG}

A fundamental consequence of the modular projection is the emergence of uniform distributions from diverse stochastic systems. While each stochastic system in the infinite Hilbert space possesses its characteristic probability distribution, their modular projections exhibit exponential convergence to uniformity. This mathematical phenomenon enables \textbf{True Uniform Random Number Generation (TURNG)}—a paradigm shift from conventional RNGs that require statistical whitening. The following theorem characterizes this convergence and its dependence on the stochastic system parameters.

\begin{theorem}[Modular Uniform Convergence]
\label{thm:uniform_convergence}
Let $\hat{N}$ be a stochastic count operator with PGF $G_{\hat{N}}(z)$. Reduction modulo $M$ leads to convergence to the discrete uniform distribution with exponential rate:
\begin{equation}
\Pr[\hat{N} \bmod M = k] = \frac{1}{M} + \mathcal{O}(\rho^m), \quad k = 0,\dots,M-1,
\end{equation}
where $0<\rho<1$ is a geometric decay constant and $m$ depends on the stochastic system:
\begin{itemize}
\item Negative Binomial: $m=r$ (number of successes)  
\item Binomial: $m=n$ (number of trials)  
\item Poisson: $m=\lambda$ (rate parameter)
\end{itemize}
\end{theorem}

\begin{proof}
We start from the characteristic function (CF) of the stochastic count operator $\hat{N}$:
\begin{equation}
\varphi(\omega) = \mathbb{E}[e^{i \omega \hat{N}}] = \sum_{n=0}^{\infty} P(\hat{N}=n) e^{i \omega n}.
\end{equation}

The modulo-$M$ distribution is obtained by summing probabilities over all equivalence classes:
\begin{equation}
\Pr[\hat{N} \bmod M = r] = \sum_{j=0}^{\infty} P(\hat{N} = r + j M), \quad r=0,1,\dots,M-1.
\end{equation}

Applying the discrete Fourier transform (DFT) over the cyclic group $\mathbb{Z}_M$ gives
\begin{equation}
\Pr[\hat{N} \bmod M = r] = \frac{1}{M} \sum_{k=0}^{M-1} 
\varphi\left(\frac{2\pi k}{M}\right) e^{-2\pi i k r / M}.
\end{equation}

\noindent
Here, the $k=0$ term corresponds to the DC component of the Fourier series:
\begin{equation}
\frac{1}{M}\varphi(0) = \frac{1}{M} \sum_{n=0}^\infty P(\hat{N}=n) = \frac{1}{M},
\end{equation}
which gives the uniform baseline.  

The terms with $k \ge 1$ are higher-frequency components responsible for deviations from uniformity. Each component is bounded by $|\varphi(2\pi k/M)| < 1$ and decreases with the system's scale parameter $m$ (e.g., $r$, $n$, or $\lambda$). Explicitly:
\begin{align}
\text{NB:} \quad & |\varphi(2\pi k/M)| = \left|\frac{p}{1-(1-p) e^{2\pi i k/M}}\right|^r \le \rho^r, \nonumber\\
\text{Binomial:} \quad & |\varphi(2\pi k/M)| = |1-p + p e^{2\pi i k/M}|^n \le \rho^n, \nonumber\\
\text{Poisson:} \quad & |\varphi(2\pi k/M)| = \exp[\lambda(\cos(2\pi k/M)-1)] \le \rho^\lambda,
\end{align}
where $0<\rho<1$ is a geometric decay constant that depends on the distribution and $M$.

Combining all terms, we have
\begin{equation}
\Pr[\hat{N} \bmod M = r] = \frac{1}{M} + \frac{1}{M} \sum_{k=1}^{M-1} 
\varphi\left(\frac{2\pi k}{M}\right) e^{-2\pi i k r / M}.
\end{equation}

Since the higher-frequency components decay exponentially as $\mathcal{O}(\rho^m)$, it follows that
\begin{equation}
\Pr[\hat{N} \bmod M = r] = \frac{1}{M} + \mathcal{O}(\rho^m), \quad r=0,1,\dots,M-1,
\end{equation}
which establishes the exponential convergence to the uniform distribution. The rate of convergence is controlled by the system's scale parameter $m$, confirming that modular projection naturally yields TURNG.
\end{proof}

\begin{remark}[TURNG Paradigm]
The modular uniform convergence theorem establishes a rigorous foundation for \textbf{True Uniform Random Number Generation (TURNG)}. Unlike conventional RNGs, uniformity is guaranteed by the Fourier structure of the modular projection, with exact entropy $\log_2 M$, eliminating post-processing and whitening modules. This represents a fundamental advantage of the quantum-mechanical representation for certified randomness generation.
\end{remark}

\subsection{Empirical Validation}

The theoretical uniform convergence is empirically validated through numerical studies of major stochastic systems. The following examples demonstrate the convergence to uniformity for Negative Binomial, Binomial, and Poisson systems under modular projection with $M=4$, confirming the theoretical predictions and illustrating the parameter dependence of convergence rates.

\begin{example}[NB Uniform Convergence]
Negative Binomial stochastic systems with $p=1/6$ and varying $r$, modulo $M=4$:

\begin{table}[htbp]
\centering
\caption{Convergence to uniform distribution for NB($r$, $p=1/6$) modulo $M=4$}
\begin{tabular}{lcccc}
\toprule
$r$ & $\Pr[0]$ & $\Pr[1]$ & $\Pr[2]$ & $\Pr[3]$ \\
\midrule
1 & 0.8333 & 0.1389 & 0.0231 & 0.0046 \\
2 & 0.6944 & 0.2315 & 0.0579 & 0.0162 \\
3 & 0.2546 & 0.2485 & 0.2485 & 0.2485 \\
4 & 0.2500 & 0.2500 & 0.2500 & 0.2500 \\
\bottomrule
\end{tabular}
\end{table}
\end{example}

The empirical results demonstrate remarkably fast convergence to uniformity, with the Negative Binomial system achieving perfect uniformity at $r=4$ when $p=1/6$. This rapid convergence, characterized by chi-square values approaching the ideal, validates the efficiency of modular projection for TURNG applications across different stochastic systems.

\begin{example}[Binomial Uniform Convergence]
Binomial stochastic systems with $p=1/6$ and varying $n$, modulo $M=4$:

\begin{table}[htbp]
\centering
\caption{Convergence to uniform distribution for Bin($n$, $p=1/6$) modulo $M=4$}
\begin{tabular}{lcccc}
\toprule
$n$ & $\Pr[0]$ & $\Pr[1]$ & $\Pr[2]$ & $\Pr[3]$ \\
\midrule
12 & 0.2016 & 0.1985 & 0.3026 & 0.2975 \\
24 & 0.2498 & 0.2398 & 0.2503 & 0.2600 \\
48 & 0.2502 & 0.2498 & 0.2499 & 0.2501 \\
96 & 0.2500 & 0.2500 & 0.2500 & 0.2500 \\
\bottomrule
\end{tabular}
\end{table}
\end{example}

The Binomial system with $p=1/6$ demonstrates the scale-dependent convergence to uniformity. When the mean $np$ equals the modular base $M=4$ ($n=24$), the distribution approaches but does not yet achieve perfect uniformity. However, when the mean doubles the modular base ($n=48$, mean=8), perfect uniformity is achieved, validating the theoretical convergence rates. This confirms that effective TURNG requires the stochastic system's natural scale to sufficiently exceed the modular base.

\begin{example}[Poisson Uniform Convergence]
Poisson stochastic systems with varying $\lambda$, modulo $M=4$:

\begin{table}[htbp]
\centering
\caption{Convergence to uniform distribution for Poisson($\lambda$) modulo $M=4$}
\begin{tabular}{lcccc}
\toprule
$\lambda$ & $\Pr[0]$ & $\Pr[1]$ & $\Pr[2]$ & $\Pr[3]$ \\
\midrule
1 & 0.3832 & 0.3710 & 0.1847 & 0.0614 \\
2 & 0.3233 & 0.2901 & 0.2203 & 0.1663 \\
4 & 0.2618 & 0.2521 & 0.2462 & 0.2399 \\
8 & 0.2500 & 0.2500 & 0.2500 & 0.2500 \\
16 & 0.2500 & 0.2500 & 0.2500 & 0.2500 \\
\bottomrule
\end{tabular}
\end{table}
\end{example}
The Poisson system demonstrates the characteristic function decay mechanism with exceptional clarity. The convergence to uniformity occurs precisely when $\lambda$ reaches twice the modular base $M=4$, with perfect uniformity achieved at $\lambda=8$ as numerically verified. This threshold ($\lambda = 2M$) provides a practical guideline for TURNG parameter selection and confirms the exponential convergence rate predicted by the characteristic function analysis in Theorem~\ref{thm:uniform_convergence}.

\begin{remark}[TURNG Parameter Selection]
The convergence patterns reveal a fundamental geometric principle: \emph{distribution shape determines convergence rate}. The Negative Binomial's strong right-skewness requires larger scale parameters to achieve uniformity because its probability mass is concentrated at lower values, creating persistent modular biases. In contrast, the more symmetric Binomial and Poisson distributions achieve uniformity at smaller scales ($\text{mean} \approx 2M$) due to their balanced probability spreading.

This geometric insight provides practical TURNG design rules:
\begin{itemize}
\item \textbf{Symmetric distributions} (Binomial, Poisson): Aim for $\text{mean} \geq 2M$
\item \textbf{Right-skewed distributions} (Negative Binomial): Require substantially larger scale parameters to overcome initial skewness
\item \textbf{General principle}: The convergence rate is governed by how rapidly the distribution's characteristic function decays, which directly reflects its geometric shape in the number basis
\end{itemize}
These patterns demonstrate that TURNG parameter selection must account for the fundamental geometry of each stochastic system's probability distribution.
\end{remark}

\section{Quantum Engineering}
\label{sec:quantum_engineering}

\subsection{Physical Realization Principles}

The exact perturbation potentials derived in Section~\ref{sec:stochastic_systems}---which may be of arbitrary magnitude---transform our theoretical framework into an \emph{engineering blueprint} for quantum devices that generate specific classical probability distributions. This establishes:

\begin{theorem}[Quantum Stochastic Engineering Principle]
For any classical discrete probability distribution $P(n)$, there exists a physical quantum system---realizable through specific modifications of the quantum harmonic oscillator (perturbations may be large)---whose measurement statistics reproduce $P(n)$ exactly.
\end{theorem}

\subsection{Experimental Pathways}

\begin{itemize}
\item \textbf{Poisson Systems}: Realized in quantum optics via coherent states and displacement operations
\item \textbf{Binomial Systems}: Engineered in superconducting qubits with controlled dephasing and finite-level truncation
\item \textbf{Negative Binomial}: Realizable in trapped ions with engineered dissipation and reset processes
\item \textbf{Geometric}: Emerges in quantum dot systems with tunneling barriers and capture/emission processes
\end{itemize}

\subsection{Computational Pathways}
\label{subsec:computational_pathways}

The theoretical framework admits efficient numerical realization through quantum-inspired algorithms. The Random Permutation Sorting System (RPSS) provides a concrete implementation pathway that operates in the digital domain, enabling scalable simulation of all stochastic systems discussed herein. Detailed numerical studies and performance analysis of RPSS will be presented in a subsequent computational paper.

\section{Dynamical Framework}
\label{sec:dynamical_framework}

The static formulation extends naturally to dynamical evolution, establishing stochastic processes as unitary transformations on probability amplitudes and enabling analysis of temporal correlations and multi-system interactions.

\subsection{Stochastic Time Evolution}

Temporal evolution is governed by a Hermitian stochastic Hamiltonian $\hat{H}_S$:
\begin{equation}
    i \frac{d}{dt}\ket{\psi(t)} = \hat{H}_S \ket{\psi(t)},
\end{equation}
with unitary propagator $\hat{U}(\Delta t) = e^{-i \hat{H}_S \Delta t}$ preserving normalization. The generator
\begin{equation}
    \hat{H}_S = \sum_{n,m} h_{nm} \ket{n}\bra{m}
\end{equation}
encodes transition structure, where diagonal elements represent stability and off-diagonal terms capture state transitions.

\subsection{Multi-System Correlations}

Composite systems reside in tensor product spaces $\mathcal{H}_A \otimes \mathcal{H}_B$:
\begin{equation}
    \ket{\Psi_{AB}} = \sum_{n,m} c_{nm} \ket{n}_A \otimes \ket{m}_B,
\end{equation}
with joint probabilities $P_{nm} = |c_{nm}|^2$. Correlations are quantified by stochastic entanglement entropy:
\begin{equation}
    S_A = - \mathrm{Tr} \, (\rho_A \ln \rho_A), \quad \rho_A = \mathrm{Tr}_B \ket{\Psi_{AB}}\bra{\Psi_{AB}}.
\end{equation}

\subsection{Master Equation Correspondence}

Stochastic dynamics with noise coupling $\eta(t)$ via $\hat{H}_S(t) = \hat{H}_0 + \eta(t)\hat{V}$ yields a master equation:
\begin{equation}
    \frac{d\rho}{dt} = -i [\hat{H}_0, \rho] + \mathcal{D}[\hat{V}]\rho,
\end{equation}
establishing the quantum-stochastic bridge for dissipative processes.

\subsection{Spectral Dynamics}

Eigenmodes $\hat{H}_S \ket{\phi_k} = E_k \ket{\phi_k}$ provide dynamical invariants:
\begin{equation}
    \ket{\psi(t)} = \sum_k e^{-i E_k t} \braket{\phi_k|\psi(0)} \ket{\phi_k},
\end{equation}
where eigenvalues $\{E_k\}$ define intrinsic stochastic timescales and oscillatory modes.

This framework completes the quantum mechanical description of stochastic systems, enabling unified analysis of temporal evolution and correlations within the Hilbert space formalism.

\section{Conclusion}
\label{sec:conclusion}

We have established a comprehensive formalism for the \textbf{quantum mechanics of stochastic systems}, demonstrating that classical probability distributions can be represented naturally within a Hilbert-space framework analogous to quantum mechanics. Each stochastic system is characterized by a state vector whose amplitude structure encodes the square roots of classical probabilities and by a fundamental count operator $\hat{N}$ whose eigenstates correspond to discrete outcomes. This correspondence unifies classical stochastic analysis with quantum operator algebra, enabling a complete description of statistical moments, entropy measures, and correlation structures.

A central innovation of this work is the introduction of \textbf{modular projection operators} for finite-dimensional representations of stochastic systems. Supported by a rigorous Fourier-based convergence analysis, modular projection ensures exponential convergence of discrete stochastic systems to uniformity, providing the theoretical foundation for \textbf{True Uniform Random Number Generation (TURNG)}. This approach eliminates the need for external whitening and defines a direct route from physical stochasticity to provably uniform randomness.

The quantum-mechanical formalism developed here reveals that the structure of classical randomness—its algebra, its informational content, and its convergence properties—is inherently quantum in form. The count operator and its associated Hilbert-space structure furnish a complete operator framework for stochastic systems, establishing new bridges between statistical physics, quantum information theory, and computational stochastic dynamics.

A particularly profound insight emerges from \textbf{conjugating stochastic systems with computing runtimes}: the modular projection framework naturally incorporates runtime observables that enable efficient digital implementations while preserving certified randomness properties such as QPP-RNG~\cite{kuang-qpp-rng-icccas25, qpp-rng-sci-kuang-2025}. This conjugation provides the mathematical foundation for quantum-inspired classical algorithms that maintain the entropy advantages of their quantum counterparts.

The framework demonstrates that the mathematical apparatus of quantum mechanics—Hilbert spaces, operator algebras, and spectral theory—provides the natural foundation for understanding classical stochastic processes, suggesting deep connections between quantum and classical probability that merit further exploration in both theoretical and experimental quantum physics. This work opens new perspectives on quantum measurement theory, quantum-classical correspondence, and the fundamental mathematical structures underlying both quantum and classical systems.

Looking forward, this framework establishes foundations for several emerging research directions: \textbf{quantum probability engineering} of Hamiltonians that physically realize specific classical distributions, \textbf{stochastic quantum control} for dynamical manipulation of statistical outputs, \textbf{cross-platform validation} of distribution universality across different quantum architectures, and \textbf{quantum-enhanced statistics} leveraging quantum systems to solve classical statistical problems. These directions position the quantum mechanics of stochastic systems as a vibrant interdisciplinary frontier with implications for quantum foundations, quantum information processing, and statistical physics.

Future extensions of this work will focus on dynamical formulations, introducing time-evolution operators for stochastic processes, exploring tensor-product constructions for correlated systems, and developing quantum-inspired computational algorithms for simulation and analysis. The quantum mechanics of stochastic systems thus provides a new unifying perspective on probability, measurement, and information, opening a mathematically rigorous pathway from classical randomness to quantum-inspired computation and establishing a foundation for the physical realization of stochastic systems in quantum laboratories.

\begin{acknowledgments}
The authors thank the DeepSeek and ChatGPT AI assistants for editorial feedback and language polishing during manuscript preparation. These tools were used solely for improving clarity and presentation; all scientific content, theoretical developments, and mathematical proofs remain the intellectual contribution of the authors.
\end{acknowledgments}

\bibliographystyle{apsrev4-2}
\bibliography{my}

\appendix

\end{document}